\documentclass[conference]{IEEEtran}

\usepackage{times,enumerate,amsmath,subfigure,graphicx,color,ifthen,amsfonts,algorithm,cite}
\usepackage{mathtools,stmaryrd,newfile,multirow,array,enumitem,cite,amsthm}
\usepackage[bookmarks,hidelinks]{hyperref}
\numberwithin{equation}{section}
\numberwithin{algorithm}{section}

\usepackage[noend]{algpseudocode}

\newtheorem{theorem}{Theorem}[section]
\newtheorem{lemma}[theorem]{Lemma}

\newtheorem{corollary}[theorem]{Corollary}

\newcommand{\lt}{\left}
\newcommand{\rt}{\right}

\newcommand{\inti}[2]{[{#1},{#2}]}

\newcommand{\h}{h}

\newcommand{\rec}[1]{\bar{#1}}
\newcommand{\org}[1]{#1}
\newcommand{\mathsmall}{\small}
\newcommand{\algsmall}{\small}
\newcommand{\thmsmall}{\small}
\newcommand{\anlsmall}{\small}

\newcommand*{\cachesize}{\ensuremath{H}}
\newcommand*{\bwcost}{\ensuremath{Q}}

\def \ALL {{\textbf{:}}}
\def \TO {\ \textbf{:}\ }
\def \CMA {\ \textbf{,}\ }

\begin{document}

\setlength{\pdfpageheight}{\paperheight}
\setlength{\pdfpagewidth}{\paperwidth}

\title{A communication-avoiding parallel algorithm for the symmetric eigenvalue problem }
\author{
\IEEEauthorblockN{Edgar Solomonik}
\IEEEauthorblockA{ETH Zurich\\
Email: solomonik@inf.ethz.ch}
\and
\IEEEauthorblockN{Grey Ballard}
\IEEEauthorblockA{Sandia National Laboratory\\
Email: gmballa@sandia.gov}
\and
\IEEEauthorblockN{James Demmel}
\IEEEauthorblockA{University of California, Berkeley\\
Email: demmel@cs.berkeley.edu}
\and
\IEEEauthorblockN{Torsten Hoefler}
\IEEEauthorblockA{ETH Zurich\\
Email: htor@inf.ethz.ch}}

\date{}
\maketitle 

\begin{abstract}
Many large-scale scientific computations require eigenvalue solvers in a scaling regime where efficiency is limited by data movement.
We introduce a parallel algorithm for computing the eigenvalues of a dense symmetric matrix, which performs asymptotically less communication than previously known approaches.
We provide analysis in the Bulk Synchronous Parallel (BSP) model with additional consideration for communication between a local memory and cache.
Given sufficient memory to store $c$ copies of the symmetric matrix, our algorithm requires $\Theta(\sqrt{c})$ less interprocessor communication than previously known algorithms, for any $c\leq p^{1/3}$ when using $p$ processors.
The algorithm first reduces the dense symmetric matrix to a banded matrix with the same eigenvalues.
Subsequently, the algorithm employs successive reduction to $O(\log p)$ thinner banded matrices.
We employ two new parallel algorithms that achieve lower communication costs for the full-to-band and band-to-band reductions.
Both of these algorithms leverage a novel QR factorization algorithm for rectangular matrices.

\end{abstract}

\section{Introduction}
\label{sec:intro}

The eigenvalue decomposition of a symmetric matrix $A$ is $A = UDU^T$ where $D$ is a diagonal matrix of eigenvalues and the columns of the orthogonal matrix $U$ are the eigenvectors of $A$.
Dense symmetric eigensolvers typically reduce the matrix to a tridiagonal matrix with the same eigenvalues, compute the eigenvalues $D$ of this tridiagonal matrix \cite{Dhillon:2006:DIM:1186785.1186788}, and, if desired, apply the orthogonal transformation backwards to compute the eigenvectors $U$.
Although algorithms for tridiagonalizing a symmetric matrix require the same asymptotic amount of work as one-sided decompositions such as LU and QR factorization,
they have a more complex dependency structure, which makes communication-efficient parallelization challenging.
Efficient execution of scientific applications such as electronic structure methods, which compute eigenvalue decompositions of a sequence of symmetric matrices (see, e.g.\ Hartree-Fock method~\cite{PSP:1733252,fock_zp_1930}), requires scalable symmetric eigensolvers.

We analyze the scalability of parallel algorithms in a Bulk Synchronous Parallel (BSP) cost model~\cite{valiant1990bridging}.
In addition to quantifying horizontal communication (data movement between processors) and synchronization, we augment the BSP model with an additional bandwidth cost parameter for vertical communication (data movement between memory and cache).
There are known algorithms for 
Cholesky, LU, and QR factorization~\cite{snirmatmul,Tiskin2007179,SD_EUROPAR_2011}, which for $n\times n$ input matrices on a $p$-processor system, have horizontal communication complexity \(W=O(n^2/\sqrt{cp})\), require \(S=O(\sqrt{cp})\) synchronizations, and use $M=O(cn^2/p)$ memory per processor.
Most commonly, 2D processor grids are used by algorithms that achieve this communication complexity for $c=1$, but 3D processor grids and more complicated schemes are needed to achieve the complexity with any $c\in [1,p^{1/3}]$ and obtain practical performance improvements~\cite{SD_EUROPAR_2011}.
For Cholesky factorization, which is simpler than LU and QR, these algorithms attain communication lower bounds $W=O\left(\frac{n^3}{pM^{1/2}}\right)$~\cite{greygeneral2010} and $W\cdot S =\Omega(n^2)$~\cite{SCKD_TECHREP_2013}, for a range of $W$ parameterized by $c$.

The best previously known algorithms for solving the symmetric eigenvalue problem directly, use 2D parallelizations and achieve the cost $W=O(n^2/\sqrt{p})$.
We introduce algorithms that reduce the horizontal communication cost asymptotically by a factor of $\sqrt{c}$, while using a factor of $c$ more memory and $\sqrt{c}$ more synchronizations, in the same fashion as previously done for one-sided factorizations.
The new algorithms are generalizations of previously known approaches, and the flexibility offered by the parameter $c$ increases the dimensionality of the tuning space for symmetric eigensolver implementations. 
In particular, employing a large $c$ is attractive for bandwidth-constrained problems on massively-parallel architectures.

Our algorithms focus on reducing the symmetric matrix to thinner and thinner banded matrices with the same eigenvalues.
This ``successive band reduction'' approach \cite{SBR1,SBR2}, i.e. reducing to an intermediate banded matrix rather than directly to tridiagonal, has been used to reduce vertical communication and synchronization costs~\cite{BDK13-TR}.
Further, in practice, algorithms using a two-stage (full-to-banded and banded-to-tridiagonal) approach~\cite{Auckenthaler2011272,haidar2011parallel} have been shown to outperform libraries that reduce directly to tridiagonal (like ScaLAPACK \cite{SCALAPACK}).
However, a disadvantage of successive band reduction is increasing the number of back transformations, which are needed to compute eigenvectors.
Unlike the forward application of transformations whose computation cost scales linearly with the matrix band-width, known algorithms for back transformations require $O(n^3)$ operations for each intermediate band-width used.

The BSP model allows us to formulate and analyze algorithms as compositions of a set of common building-blocks.
We leverage algorithms for matrix multiplication and QR factorization within our symmetric eigensolvers.
For QR factorization, we provide an approach that extends approaches for tall-and-skinny matrices~\cite{demmel:A206} and square matrices~\cite{Tiskin2007179} to be efficient for arbitrary rectangular matrices.

We use these building blocks to define algorithms for reducing a dense matrix to a banded matrix, and a banded matrix to a thinner band-width, while preserving eigenvalues.
Our main algorithm combines these, using $O(\log p)$ intermediate band-widths.
The algorithm is work-efficient for computing eigenvalues, requires $O(n^2/\sqrt{cp})$ horizontal communication, $O(n^2\log p/\sqrt{cp})$ vertical communication, and $O(\sqrt{cp} \log^2 p)$ synchronizations (BSP supersteps).
Known approaches for back-transformations to compute eigenvectors require the same asymptotic amount of computation for matrices of any band-width, meaning our approach may require a computation cost of $O(n^3\log p/p)$ if all eigenvectors are needed.
We leave the analysis of back-transformation computation for future work, but propose a potential approach to reduce the number of intermediate band-widths needed by our symmetric eigensolver.

\section{Theoretical Cost Model}
\label{sec:model}

We use the Bulk Synchronous Parallel (BSP) model~\cite{valiant1990bridging} with an additional parameter to measure the cost of traffic between memory and cache.
We derive asymptotic bounds on the parallel running-time of our algorithms for this two-level architectural model, with consideration for both communication between processors and in the memory hierarchy of each processor.
The BSP model permits an all-to-all communication to be done with unit synchronization cost, which will allow us to construct BSP algorithms for general matrix distributions and compose them without significant overhead.

We employ cost notation typically used for the $\alpha$--$\beta$ communication model.
As all stored and communicated datasets in this paper consist exclusively of floating-point numbers, we quantify sizes in terms of `words' (floating-point numbers of a given precision).
We model the memory hierarchy of each processor by a main `slow' memory (i.e.\ DRAM) and a `fast' memory (i.e.\ cache).
We permit interprocessor (horizontal) communication to move data between main memories of different processors, and intraprocessor (vertical) communication to move data between main memory and cache of a single processor.
Our architectural model is characterized by the following parameters:
\begin{itemize}
\item $p$ -- processors on a fully-connected network,
\item $M$ -- words of memory owned by each processor,
\item $H$ -- words of cache owned by each processor,
\item $\gamma$ -- time to compute a floating point operation,
\item $\beta$ -- time to send or receive a word,
\item $\nu$ -- time to move a word between cache and memory,
\item $\alpha$ -- time to perform a (global) synchronization.
\end{itemize}
We bound the cost of each algorithm by measuring four quantities:
\begin{itemize}
\item $F$ -- number of local floating point operations performed (computation cost),
\item $W$ -- number of words of data moved between processors (horizontal communication cost),
\item $\bwcost$ -- number of words of data moved between main memory and cache (vertical communication cost),
\item $S$ -- number of BSP supersteps (synchronization cost).
\end{itemize}
If at each superstep $i\in \inti 1S$, processor $j$ performs $F_i^j$ local operations, sends and receives $W_i^j$ total words, and performs $\bwcost_i^j$ reads and writes to memory, then the costs of the BSP algorithm are
{\anlsmall \[F=\sum_{i=1}^S\max_{j\in\inti 1p} F_i^j, \ \  
W=\sum_{i=1}^S\max_{j\in\inti 1p} W_i^j, \ \ 
\bwcost=\sum_{i=1}^S\max_{j\in\inti 1p} \bwcost_i^j,\]}
and the BSP execution time of this algorithm is
{\anlsmall $$T = \Theta(\gamma \cdot F + \beta \cdot W + \nu \cdot \bwcost + \alpha \cdot S).$$}
This model does not consider overlap between communication and computation (or between other costs), as such overlap does not affect the overall asymptotic time.

We simplify asymptotic cost expressions by assuming $\gamma\leq \beta$.
Further, we write only vertical communication terms which are not associated with horizontal communication or with computations that achieve a factor of $\sqrt{H}$ cache reuse (optimal for matrix multiplication \cite{Jia-Wei:1981:ICR:800076.802486}).
These simplifications correspond to the assumptions on the relative communication times, $\nu \leq \beta$
and the floating point rate $\nu \leq \gamma\cdot \sqrt{H}$.
However, general vertical communication cost upper-bounds may be obtained from our stated results for arbitrary $\nu$ by reinserting the term $O(\nu\cdot (F/\sqrt{H} +W))$.

We will provide asymptotic bounds for the BSP cost of all algorithms in the paper.
Sometimes, we will employ algorithms as building blocks whose cost has been analyzed in the standard $\alpha-\beta$ model, which is restricted to point-to-point messaging (pairwise synchronization).
These algorithms are trivially translated to the BSP model used in this paper, which is less restrictive (allows bulk synchronizations).

Throughout the paper, we will assume that matrix dimensions are greater than and divisible by the number of processors.
When it is clear that the asymptotic costs would not be affected, we will also omit floors and ceilings when subdividing the number of processors and matrix dimensions.

\section{Building Blocks}
\label{sec:prevwork}

We first state known results and provide minor extensions to quantify the complexity of matrix multiplication and of QR factorization in our cost model.
These results will be critical in the cost analysis of the new symmetric eigensolvers, which use matrix multiplication and QR factorization as subroutines.

  \subsection{Matrix Multiplication}
  \label{sec:prevwork:mm}

Our symmetric eigensolvers will perform matrix multiplications, often of nonsquare matrices.
We consider the BSP cost of multiplication of arbitrary rectangular matrices with any starting distribution.
Additionally, we specially consider the BSP cost of a matrix multiplication of a pre-replicated matrix with another matrix in an arbitrary distribution.
We start with the vertical communication cost of a matrix multiplication done by a single processor.
\begin{lemma}
\label{lem:seqMM}
The multiplication of matrices of dimensions $m\times n$ and $n\times k$ can be done by a single processor in time,
{\thmsmall 
\begin{align*}
O(\gamma\cdot mnk+ \nu\cdot\left[mn+mk+nk\right]).
\end{align*}
}
\end{lemma}
The Rec-Mult algorithm \cite[Theorem 1]{FLPR99} obtains the vertical communication cost given in Lemma~\ref{lem:seqMM}.
We omit the usual term $O(\nu\cdot mnk/\sqrt{H})$, since we have $\nu \leq \gamma\cdot \sqrt{H}$.

We now consider the full BSP cost of parallel rectangular matrix multiplication. 
The communication cost of square matrix multiplication is well known~\cite{mccol_tiskin_99,dekel:657,matmul3d,snirmatmul,berntsen1989communication,Johnsson:1993:MCT:176639.176642}.
The horizontal costs of rectangular matrix multiplication have also been analyzed within the $\alpha$--$\beta$ communication model, where a recursive algorithm was proposed \cite{demmel2013communication} that attains the communication lower bound.
We show that the algorithm in \cite{demmel2013communication} can be executed within the time specified in the subsequent Lemma, for any initial load balanced distribution of the matrices.
It is possible to also design different matrix multiplication algorithms in the BSP model with a $\Theta(\log p)$ factor less in synchronization cost, but the overall synchronization costs of our QR and symmetric eigensolve algorithms (which use the subsequent Lemma) would not be affected.
We parameterize the memory used by the algorithm by a parameter $v$, which controls how many block matrix multiplications are performed by each processor.
\begin{lemma}
\label{lem:parMM}
For any $v\geq 1$, the multiplication of matrices of dimensions $m\times n$ and $n\times k$ in any load-balanced starting layout can be done in BSP time,
{\thmsmall
\begin{align*}
O\bigg(&\gamma\cdot \frac{mnk}{p}  + \beta\cdot \bigg[\frac{mn+nk+mk}{p}  \\
&+v^{1/3}\bigg(\frac{mnk}{p}\bigg)^{2/3}\bigg]    
+ \alpha \cdot v \log p\bigg),
    \end{align*}
}
using $M=O\big(\frac{mn+nk+mk}p+\big(\frac{mnk}{vp}\big)^{2/3}\big)$ memory.
\end{lemma}
\begin{proof}
We consider the cost of the recursive `CARMA' algorithm \cite{demmel2013communication}.
The algorithm assumes specific initial matrix layouts, but does not assume any initial data is replicated.
Therefore, starting from load balanced layouts, the BSP time to move to the layouts specified by CARMA is $O(\beta\cdot \frac{mn+nk+mk}{p}+\alpha)$.
Because the computation is load balanced, the computation cost is $O(\gamma\cdot mnk/p)$.
The latency cost of the CARMA algorithm is an upper-bound on the number of BSP supersteps necessary to execute it.
In \cite{demmel2013communication}, the latency cost is shown to be $O\big(\frac{mnk}{pM^{3/2}}\log p\big)=O(v\log p)$.
The communication cost of CARMA is presented in cases for 1D, 2D, and 3D processor grids.
We show that the postulated BSP time upper-bound holds for all cases.

We first argue that the vertical communication cost of the local matrix multiplications (given by Lemma~\ref{lem:seqMM}) is dominated by horizontal communication due to the assumption $\beta \geq \nu$.
In the 3D case, the operand matrix blocks are nearly square, and either one of the operands or the output is always communicated, so horizontal communication cost dominates vertical communication cost.
In the 1D and 2D cases, each processor performs a single local matrix multiplication, where the largest operand has size $O(\frac{mn+nk+mk}p)$, since it is the local block of the largest matrix, which is distributed across all processors. 

For the horizontal communication costs, we let $d_1=\min(m,n,k)$, $d_2=\mathrm{median}(m,n,k)$, and $d_3=\max(m,n,k)$ as in \cite{demmel2013communication}.
If $p< d_3/d_2$ (1D case), then $d_1d_2<d_1d_3/p$, so the provided cost $O(\beta\cdot d_1d_2)=O(\beta\cdot (mn+nk+mk)/p)$.
If $d_3/d_2 \leq p \leq d_2d_3/d_1^2$ (2D case), then the provided cost $O(\beta\cdot \sqrt{d_1^2d_2d_3/p})=O(\beta\cdot (mn+nk+mk)/p)$.
Finally, if $p>d_2d_3/d_1^2$ (3D case), the provided cost $O(\beta \cdot [mnk/(p\sqrt M)+(mnk/p)^{2/3}])=O(\beta\cdot v^{1/3}(mnk/p)^{2/3})$.

\end{proof}

The algorithm analyzed in Lemma~\ref{lem:parMM} allows any initial load balanced matrix distributions.
We now consider Algorithm~\ref{alg:StMM}, which assumes an initial distribution with replicated data and subsequently can multiply certain matrices in less time than given by Lemma~\ref{lem:parMM}.
In Algorithm~\ref{alg:StMM}, one of the input matrices is stored redundantly on $c=p^{2\delta-1}$ 2D processor grids for any $c\in[1,p^{1/3}]$ ($\delta\in[1/2,2/3]$).
The parameterization by $\delta$ is the same as $\alpha$ in~\cite{Tiskin2007179}, while $c$ is the same replication factor as in~\cite{SD_EUROPAR_2011}.
The parameter $w$ controls the number of supersteps (block matrix multiplications) in Algorithm~\ref{alg:StMM}.

The algorithm permits the distribution to be defined as a blocking of the matrices after permutation by $P^{(1)},P^{(2)}$.
Our analysis assumes the blocking is roughly, but not necessarily exactly load balanced, permitting the analysis to be used within a cyclic or block-cylic matrix factorization algorithm where different processors perform updates (matrix multiplications) with a slightly different amount of local data at each step.
We will employ Algorithm~\ref{alg:StMM} with cyclic distributions, for which $P^{(1)}_{ij}=1$ for $i=(j\bmod q)(m/q)+\lfloor j/q\rfloor$ and $P^{(2)}_{jk}=1$ for $k=(j\bmod q)(n/q)+\lfloor j/q\rfloor$.
On each processor grid layer, the algorithm executes a variant of the SUMMA algorithm~\cite{Geijn:SUMMA:1997}, which communicates the operand $B$ and reduces the output $C$.
This variant is chosen, since we will use the algorithm in cases where the operand $A$ is of greater size than $B$ and $C$.

\begin{lemma}\label{lem:repupd}
Consider Algorithm~\ref{alg:StMM} for multiplication of matrices $A$ and $B$ of dimensions $m\times n$ and $n\times k$, where the initial distributions of $A$ and $B$ satisfy the stated requirements for permutations $P^{(1)}$ and $P^{(2)}$ where each block $A_{ij}$ of $P^{(1)}AP^{(2)}$ has dimensions $O(m/p^{1-\delta})\times O(n/p^{1-\delta})$.
Then, using $M=O(mn/p^{2(1-\delta)}+(mk+nk)/(wp^{\delta}))$ memory for any $w\in\inti 1{p^{1-\delta}}$, the algorithm can be executed in BSP time,
{\thmsmall \begin{align*}
O\bigg(\gamma\cdot \frac{mnk}{p}  + \beta\cdot \frac{mk+nk}{p^\delta} + \alpha\cdot w\bigg),
\end{align*}
}
when $\cachesize\geq mn/p^{2(1-\delta)}$ and the copies of $A$ start inside cache, 
and otherwise with an extra cost of $O(\nu\cdot\frac{wmn}{p^{2(1-\delta)}})$. 
\end{lemma}
\begin{algorithm}
\caption {$[C]\gets \text{Streaming-MM}(A,B,P^{(1)},P^{(2)},\Pi)$}
\label{alg:StMM}
{\algsmall
\begin{algorithmic}[1]
\Require{Given positive integers $p,m,n,k,w$ and $\delta\in[1/2,2/3]$: $\Pi$ is a grid of $q\times q\times c$ processors with $q=p^{1-\delta}$ and $c=p^{2\delta-1}$, $A$ is $m\times n$, $B$ is $n\times k$. 
For each $l\in\inti 1c$, $\Pi[i,j,l]$ owns all elements in $A_{ij}$, defined by square permutation matrices {\small $P^{(1)}$, $P^{(2)}$, as
  \(P^{(1)}AP^{(2)} = \begin{bmatrix} A_{11} & \cdots & A_{1q} \\
                            \vdots & \ddots  & \vdots \\
                            A_{q1} & \cdots & A_{qq} \end{bmatrix} \).} \\
$B$ is in any load balanced layout over all $p$ processors. }  \State Let $z=wc$ 
  \State Partition $B$ into blocks:
{\small  \(P^{(2)}{}^TB = \begin{bmatrix} B_{11} & \cdots & B_{1z} \\
                           \vdots &   & \vdots \\
                           B_{q1} & \cdots & B_{qz} \end{bmatrix} \).}
  \State Redistribute $B$ so that each $\Pi[i,j,l]$ owns $k/(zq)$ columns of $B_{jh}$ for each  $h \in \{l,l+c,\ldots, l+(w-1)c\}$. \label{li:sm_redist}
  \Comment {Execute loop iterations in parallel}
  \For {$i\in \inti 1q, j \in \inti 1q, l \in \inti 1c$}
    \Comment {Execute loop iterations in sequence}
    \For {$h \in \{l,l+c,\ldots, l+(w-1)c\}$}
      \State Gather $B_{jh}$ on $\Pi[i,j,l]$ \label{alg:StMM:gather}
      \State Compute $\bar{C}_{ijh}=A_{ij}\cdot B_{jh}$ on $\Pi[i,j,l]$ \label{alg:StMM:compute}
      \State Reduce-scatter $C_{ih} = \sum_{j=1}^c \bar{C}_{ijh}$ so that each $\Pi[i,j,l]$ owns $k/(zq)$ columns of $C_{ih}$ \label{alg:StMM:reduce}
    \EndFor
  \EndFor
  \Require{$C=A\cdot B$ is distributed so that each processor in $\Pi$ owns $mk/p$ elements of $C$.}
\end{algorithmic}
}
\end{algorithm}

\begin{proof}
As required by Algorithm~\ref{alg:StMM}, $B$ starts in any load-balanced distribution over the $p$ processors.
As the initial layout is load-balanced the redistribution done on line~\ref{li:sm_redist} costs $O(\beta\cdot nk/p +\alpha)$.
The gather on line~\ref{alg:StMM:gather} and reduce-scatter on line~\ref{alg:StMM:reduce} are dual communication patterns with respect to each other.
Together, they cost $O(\beta\cdot (mk+nk)/(qcw)+\alpha)$, and over all $w$ iterations over index $h$ cost, $O(\beta\cdot (mk+nk)/(qc) + \alpha\cdot w)=O(\beta\cdot (mk+nk)/p^\delta +\alpha\cdot w)$.

The $w$ local matrix multiplications take time,
{\mathsmall \[O\left(\gamma\cdot \frac{mnk}{p}  + \nu\cdot\bigg(\frac{wmn}{p^{2(1-\delta)}}+\frac{mk+nk}{p^\delta}\bigg) \right),\]}
by Lemma~\ref{lem:seqMM}. However, if the entire matrix $A$ starts in cache, which is possible if
$\cachesize\geq mn/p^{2(1-\delta)}$, it suffices to read only the entries of $B_{jh}$ from memory
into cache and write the entries of $\bar{C}_{ijh}$ out to memory. In this case, the vertical communication cost is
 $O(\nu\cdot \frac{mk+nk}{qc})=O(\nu \cdot \frac{mk+nk}{p^\delta}).$
This term is dominated by the interprocessor communication term since $\beta \geq \nu$.
The memory usage corresponds to the storage necessary for each block: $A_{ij}$, $B_{jh}$, and $\bar{C}_{ijh}$, $M=O\left(\frac{mn}{p^{2(1-\delta)}} + \frac{mk+nk}{wp^\delta}\right)$.
\end{proof}

  \subsection{QR Factorization}
  \label{sec:prevwork:qr}

We will use QR factorization within our symmetric eigensolver algorithms to obtain orthogonal transformations that introduce zeros when applied to the symmetric matrix.
The vertical communication cost of executing a sequential QR factorization is proportional to that of matrix multiplication.
\begin{lemma}
\label{lem:seqqr}
The QR factorization of an $m\times n$ matrix $A$ with $m\geq n$ can be done by a single processor in time, 
{\thmsmall $$O( \gamma\cdot mn^2 + \nu \cdot mn).$$}
\end{lemma}
The sequential Communication-Avoiding QR (CAQR) algorithm achieves the vertical communication cost given above \cite{demmel:A206}.
The Householder representation, lower trapezoidal $m\times n$ matrix $U$ and upper-triangular $n\times n$ matrix $T$ so that $Q=I-UTU^T$, may be obtained with the cost of Lemma~\ref{lem:seqqr} using Householder reconstruction~\cite{BDGJNS_IPDPS_2014}.

We now consider parallel QR factorization, firstly for square matrices.
\begin{lemma}
\label{lem:parQRsq}
The QR factorization of an $n\times n$ matrix $A$ distributed in any load-balanced layout can be computed using $M=O\big(\frac{n^2}{p^{2(1-\delta)}}\big)$ memory for any $\delta\in[1/2,2/3]$ in BSP time,
{\thmsmall \begin{align*}
O\bigg(&\gamma\cdot \frac{n^3}{p}
            + \beta\cdot \frac{n^2}{p^\delta} 
            + \alpha\cdot p^\delta \bigg).
\end{align*}
}
\end{lemma}
The QR algorithm given by~\cite{Tiskin2007179} in the BSP model achieves the costs given in Lemma~\ref{lem:parQRsq}.
The vertical communication cost was not analyzed in~\cite{Tiskin2007179}.
However, the algorithm consists purely of distributed matrix multiplications or QR factorizations, which by Lemma~\ref{lem:seqMM} and Lemma~\ref{lem:seqqr} have a vertical communication cost proportional to the matrix sizes.
As the analysis in~\cite{Tiskin2007179} assumes all matrices that participate in multiplication or QR factorization are communicated, due to $\nu < \beta$, the horizontal communication cost dominates the vertical communication costs associated with these operations.

We now adapt the QR algorithm from~\cite{Tiskin2007179} to handle rectangular matrices with a desirable asymptotic cost (the embedding used in~\cite{Tiskin2007179} is inefficient for tall-and-skinny matrices).
Our adaptation is based on a binary QR reduction tree, with QR factorizations of nearly square matrices done at every node in the tree performed using the algorithm from~\cite{Tiskin2007179}.
An approach employing a QR reduction tree using Givens rotations goes back to~\cite{golub1986parallel}, a blocked flat tree approach (optimal sequentially) was presented in~\cite{Gunter:2005:POC:1055531.1055534}, and a parallel block reduction tree approach was given earlier in~\cite{da2002new}.
Our approach is closest to the TSQR algorithm~\cite{demmel:A206}, except a set of up to $q_\text{max}$ processors works on each tree node.

Algorithm~\ref{alg:rectQR} computes the QR factorization of an $\org{m}\times n$ matrix, outputting the first $n$ columns of the orthogonal $Q$ factor, as well as the $n\times n$ upper-triangular matrix $R$.
The algorithm assumes the existence of a sequential routine `QR' and a parallel routine for (nearly) square matrices `square-QR'.
\begin{theorem}
\label{thm:parQRrect}
Algorithm~\ref{alg:rectQR} can compute the QR factorization of any $\org{m}\times n$ matrix $A$ with $\org{m}\geq n$ in a load-balanced layout,
using $M=O\big(\big(\frac{n^\delta\org{m}^{1-\delta}}{p^{1-\delta}}\big)^2\big)$ memory for any $\delta\in[1/2,2/3]$, in BSP time, 
\thmsmall{
\begin{align*}
O\bigg(\gamma\cdot \frac{\org{m}n^2}{p}
            + \beta\cdot \left(\frac{\org{m}^\delta n^{2-\delta}}{p^\delta}+ \frac{\org{m}n}{p}\right)
            + \alpha\cdot \bigg(\frac{np}{\org{m}}\bigg)^\delta\log^2{p}
\bigg).
\end{align*}
}
\end{theorem}

\begin{proof}
We assume without loss of generality that $\org{m}/n$ and $\org{p}$ are powers of two.
Let $T(\rec{m})$ be the cost of Algorithm~\ref{alg:rectQR} for an $\rec{m}\times n$ matrix using $p$ processors.
Note that $\org{m}$ corresponds to the number of rows in the original input matrix, while $\rec{m}$ will be used to refer to the number of rows at a given recursive step.
We select the maximum number of processors to be used in base-case square QR factorizations to be $q_\text{max}=\frac{\org{p}\org{n}}{\org{m}}\log(\org{p})^{1/\delta}$, in order to minimize synchronization cost while achieving an optimal horizontal communication cost.

The cost of the sequential base case of Algorithm~\ref{alg:rectQR} is, by Lemma~\ref{lem:seqqr},
\(T_\text{bs1}(\rec{m})=O( \gamma\cdot \rec{m}n^2 + \nu \cdot \rec{m}n).\)
When reaching the square base case (dimension $2n\times n$, since $\org{m}/n$ is a power of two), we employ the square QR algorithm~\cite{Tiskin2007179} 
with up to $q_\text{max}=\frac{\org{p}n}{\org{m}}\log(\org{p})^{1/\delta}$ processors.
We can bound the cost of this QR is by Lemma~\ref{lem:parQRsq}.
We break the cost into two cases: $T_\text{bp}(\rec{p})=T_\text{bp1}(\rec{p})$ when $\rec{p}<q_\text{max}$ and
$T_\text{bp}(\rec{p})=T_\text{bp2}$ when $\rec{p}\geq q_\text{max}$, where
{\mathsmall
\begin{align*}
&T_\text{bp1}(\rec{p})=O( \gamma\cdot n^3/\rec{p} + \beta \cdot n^2/\rec{p}^\delta + \alpha\cdot \rec{p}^\delta), \\
&T_\text{bp2}=O\bigg( \gamma\cdot \frac{\org{m}n^2}{\org{p}\log(\org{p})^{1/\delta}} 
+ \beta \cdot \frac{\org{m}^\delta n^{2-\delta}}{\org{p}^\delta \log \org{p}} 
+\alpha\cdot \Big(\frac{n\org{p}}{\org{m}}\Big)^\delta\log \org{p}\bigg).
\end{align*}
}
The square QR algorithm requires that the matrix be embedded into a slanted panel~\cite{Tiskin2007179}.
This can be done generally by using a somewhat larger matrix, but in all except the first recursive call, the $2b\times b$ matrix will have the structure of two stacked upper-triangular matrices.
The rows of these upper-triangular matrices can be interleaved to produce a slanted panel without embedding into a larger matrix.

\begin{algorithm}[t]
\caption {$[Q,R]\gets \text{rect-QR}(A,\Pi)$}
\label{alg:rectQR}
{\algsmall
\begin{algorithmic}[1]
\Require{Given positive integers $p,m,n,q_\text{max}$ and $\delta\in[1/2,2/3]$: $\Pi$ is a set of $p$ processors, $A$ is $m\times n$, $m/n$ and $p$ are powers of two, and each $\Pi[i]$ owns $mn/p$ elements of $A$.}
  \If {$p=1$} Compute $[Q,R]=\text{QR}(A)$ sequentially and exit. \EndIf
  \If {$m\leq 2n$} Compute $[Q,R]=\text{square-QR}(A,\Pi[1:\min(p,q_\text{max})])$ and exit. \EndIf
  \State Let $r=\min(p,\lceil \frac{m}{2n}\rceil )$ and partition $A=\begin{bmatrix} A^T_1 & \cdots & A^T_r \end{bmatrix}^T$ so that each $A_i$ is $m/r\times n$
  \Comment {Execute loop iterations in parallel}
  \For {$i \in \inti 1r$}
    \State $[W_i,R_i]=\text{rect-QR}(A_i,\Pi[(i{-}1)(p/r){+}1 \TO i(p/r)])$ \label{li:recSQ}
  \EndFor
  \State $[Z,R]=\text{rect-QR}\lt(\begin{bmatrix} R^T_1 & \cdots & R^T_r \end{bmatrix}^T,\Pi\rt)$ \label{li:recmain}
  \State Partition $Z=\begin{bmatrix} Z^T_1 & \cdots & Z^T_r \end{bmatrix}^T$ so that each $Z_i$ is $n\times n$
  \Comment {Execute loop iterations in parallel}
  \For {$i \in \inti 1r$}
    \State Compute $Q_i=W_iZ_i$ using $\Pi[(i{-}1)(p/r)+1 \TO i(p/r)]$ \label{li:recQRMM}
  \EndFor
  \Require{$A=Q\cdot R$ where $Q=\begin{bmatrix} Q^T_1 & \cdots & Q^T_r \end{bmatrix}^T$ is $m\times n$ with orthogonal columns, $R$ is $n\times n$ and upper-triangular, both are distributed in load balanced layouts across $\Pi$.}
\end{algorithmic}
}
\end{algorithm}

The recursive calls on line~\ref{li:recSQ} always immediately encounter one of the base-cases.
The only time base cases can have a matrix with dimension other than $2n\times n$ is during the invocations on line~\ref{li:recSQ} at the first recursive step of the algorithm, and only when $\org{m}>2n\org{p}$.
Therefore, we consider this first recursive step of Algorithm~\ref{alg:rectQR} separately. The cost of the first recursive step, when $\org{m}>2n\org{p}$, includes 
\begin{itemize}
\item the cost of a potential redistribution, $O(\beta\cdot \org{m}n/\org{p}+\alpha)$,
\item the cost of the invocations on line~\ref{li:recSQ} (which lead to base cases), $T_\text{bs1}(\org{m}/\org{p})$, since $r=\min(\org{p},\lceil \org{m}/2n\rceil)=\org{p}$,
\item the cost of the matrix multiplications on line~\ref{li:recQRMM}, which are done concurrently, each by a single processor, is
$O(\gamma \cdot \org{m}n^2/\org{p} + \nu\cdot \org{m}n/\org{p})$.
\end{itemize}
We can therefore bound the total BSP time of the algorithm for $\org m>2n\org p$ by
{\mathsmall
\begin{align*}
T(\org{m})  &= T(n\org{p}) +T_\text{bs1}\Big(\frac{\org{m}}{\org{p}}\Big) 
+ O\Big(\gamma \cdot \frac{\org{m}n^2}{\org{p}} + \beta\cdot \frac{\org{m}n}{\org{p}}\Big) \\
&= T(n\org{p}) +O(\gamma \cdot \org{m}n^2/\org{p} + 
\beta\cdot \org{m}n/\org{p}+\alpha). \end{align*}
}
We note that the cost of this first recursive step for $\org{m}>2n\org{p}$ is no greater than the cost postulated in the theorem.
We now focus on subsequent recursive calls into line~\ref{li:recQRMM} or the case when $\org{m}\leq 2n\org{p}$, the matrix multiplications done on line~\ref{li:recQRMM} involve matrices of size at most $2n\times n$, each executed using $\org{p}n/\rec{m}$ processors.
By Lemma~\ref{lem:parMM} with $v=(\org{p}n/\rec{m})^{2-3\delta}$, these matrix multiplications (done concurrently) take time, $T_\text{MM}(\rec{m})=$
{\mathsmall \[O\bigg(\gamma \cdot \frac{\rec{m}n^2}{\org{p}} + \beta\cdot \bigg(\frac{\rec{m}n}{\org{p}} + \frac{\rec{m}^\delta n^{2-\delta}}{\org{p}^{\delta}}\bigg)  
+ \alpha\cdot \bigg(\frac{\org{p}n}{\rec{m}}\bigg)^{2-3\delta}\log{\org{p}}\bigg)\] }
and use $M=O\big(\big(\frac{n^\delta\rec{m}^{1-\delta}}{\org{p}^{1-\delta}}\log{\org{p}}\big)^2\big)$ memory.
When combined with the concurrent recursive calls on line~\ref{li:recQRMM} on matrices of size $2n\times n$ with $\org{p}n/\rec{m}$ processors and the recursive call on line~\ref{li:recmain} on a matrix of size $\rec{m}/2\times n$ with all $\org{p}$ processors, we obtain the following BSP time recurrence for $\rec m\leq 2n\org p$,
{\mathsmall
\begin{align*}
T(\rec{m}) =& T(\rec{m}/2) +T_\text{bp}(\org{p}n/\rec{m}) + T_\text{MM}(\rec{m}),
\end{align*}
}
where $T_b(\org{p}n/\rec{m})$ is a base case where up to $q_\text{max}$ processors perform the QR.
We consider the two cases (for $\rec m\leq 2n\org p$),
{\mathsmall
\begin{align*}
T(\rec{m}) =& T(\rec{m}/2) + T_\text{MM}(\rec{m})
+
\begin{cases}
T_\text{bp1}(\org{p}n/\rec{m}) &: \org{p}n/\rec{m} < q_\text{max} \\
T_\text{bp2}  &: \org{p}n/\rec{m} \geq q_\text{max}  
\end{cases}
\end{align*}
}
Since $q_\text{max}=\frac{\org{p}n}{\org{m}}\log(\org{p})^{1/\delta}$, and $\rec{m}$ decreases by a factor of two at each step, up to the first $(1/\delta)\log \log p$ recursive steps make the call on line~\ref{li:recSQ} with more than $q_\text{max}$ processors.
The computation and communication cost of these calls are no greater than that of matrix multiplication (part of $T_\text{MM}(\rec{m})$), while the synchronization cost increases geometrically, going up to the latency cost in $T_\text{bp2}$.
Therefore, the recurrence is asymptotically equivalent to (for $\rec m\leq2n\org p$),
{\mathsmall
\begin{align*}
T&(\rec{m}) = T(\rec{m}/2) + T_\text{MM}(\rec{m}) + T_\text{bp2} \\
=&T(\rec{m}/2) + O\bigg(\gamma\cdot \bigg(\frac{\rec{m}n^2}{\org{p}}+\frac{\org{m}n^2}{\org{p}\log(\org{p})^{1/\delta}}\bigg) \\
&+ \beta\cdot \bigg(\frac{\rec{m}n}{\org{p}} + \frac{\rec{m}^\delta n^{2-\delta}}{\org{p}^{\delta}}+ \frac{\org{m}^\delta n^{2-\delta}}{\org{p}^\delta \log \org{p}}\bigg) 
+ \alpha\cdot \bigg(\frac{\org{p}n}{\org{m}}\bigg)^{\delta}\log p\bigg). 
\end{align*}
}
Since, $\rec{m}\leq 2n\org p$, one of the base-cases is reached after $\log p$ steps, and so the above time reduces to the one postulated in the theorem.
\end{proof}
Alternate communication-efficient formulations of a rectangular QR algorithm are also possible (for instance by combining column-recursion~\cite{EG98} with communication-efficient matrix multiplication, see~\cite{ES_dissertation_2014}).
We would like to work with the Householder representation to apply orthogonal transformations efficiently in our symmetric eigensolver algorithms, so we give the following corollary.
\begin{corollary}
\label{cor:parQRrectY}
The Householder representation of the $\org{m}\times n$ orthogonal matrix $Q$ computed by Algorithm~\ref{alg:rectQR}, $Q=(I-UTU_1^T)$, where $U_1$ is the lower triangular top $n\times n$ block of $U$, while $T$ is upper-triangular and $U^TU=T^{-1}+T^{-T}$, can be obtained with no greater asymptotic cost or memory than given in Theorem~\ref{thm:parQRrect}.
\end{corollary}
\begin{proof}
The Householder representation $U,T$ can be obtained stably by executing $[U_1,W_1]=\text{LU}(Q_1-S)$ where $Q_1$ is the top $n\times n$ block of $Q$ and $S$ is a diagonal sign matrix, then computing $U=QW_1^{-1}$ and $T=W_1U_1^{-T}$~\cite{BDGJNS_IPDPS_2014}.
The matrices $U_1$, $W_1$, $U_1^{-1}$, and $W_1^{-1}$ can be obtained by a parallel non-pivoted LU factorization algorithm augmented to subtract $S$ as in~\cite{BDGJNS_IPDPS_2014}, which makes the matrix diagonally dominant.
The LU algorithms in~\cite{Tiskin_LU_2002} and~\cite{SD_EUROPAR_2011} would both obtain the desired costs, but the former is slightly more convenient for our analysis.

When executed using $pn/\org{m}$ processors, the algorithm in~\cite{Tiskin_LU_2002} takes BSP time, $O(\gamma \cdot \org{m}n^2/p + \beta \cdot \org{m}^\delta n^{2-\delta}/p^\delta + \alpha \cdot (np/\org{m})^\delta)$.
This cost was presented in~\cite{Tiskin_LU_2002}, modulo analysis of vertical communication cost, but as the algorithm is based purely on parallel multiplication of square matrices, the vertical communication cost is dominated by the horizontal communication cost.
The algorithm also outputs the inverses of the triangular factors~\cite{Tiskin_LU_2002}, so
matrix multiplications suffice to compute $U=QW_1^{-1}$ and $T=W_1U_1^{-T}$.
These can be done using all the processors in time, $O(\gamma \cdot \org{m}n^2/p + \beta \cdot \org{m}^\delta n^{2-\delta}/p^\delta + \alpha)$ with $M=O\big(\big(\frac{n^\delta\org{m}^{1-\delta}}{p^{1-\delta}}\big)^2\big)$ memory.
As these costs and memory usage are no greater than in Theorem~\ref{thm:parQRrect}, we arrive at the postulated conclusion.
\end{proof}

\section{Symmetric Eigensolvers}
\label{sec:se}

Algorithms for blocked computation of the eigenvalue decomposition of a symmetric matrix via
a tridiagonal matrix were studied by~\cite{dongarra1989block,Dongarra1992973,Joffrain:2006:AHT:1141885.1141886}. 
These algorithms reduce an $n\times n$ symmetric matrix $A$ to a matrix $B$ with band-width $b$ and the same eigenvalues as $A$ via a series
of $k=(n-b)/b$ orthogonal transformations,
{\anlsmall $$B=Q_1^T\cdots Q_k^TBQ_k\cdots Q_1,$$}
where each $Q_i$ is representable in terms of $b$ Householder vectors, aggregated in a trapezoidal matrix $U_i$, as $Q_i=(I-U_iT_iU_i^T)$. 

A key property employed by these algorithms is that each two-sided trailing matrix update of blocked Householder
transformations may be done as a rank-$2b$ symmetric update.
To compute the two-sided transformation $Q^TXQ$ where $X=X^T$ and $Q=(I-UTU^T)$, we can write
{\anlsmall 
\begin{align}
Q^TXQ=& (I-UT^TU^T)X(I-UTU^T) \nonumber \\
=& X+UV^T+VU^T \label{eq:auvvu},
\end{align}
}
where $V= \frac{1}{2}UT^TU^TXUT-XUT$. This form of the update is cheaper to compute than the explicit two-sided update and is easy to aggregate by appending
additional vectors to $U$ (to aggregate the Householder form itself requires computing a larger $T$ matrix). 
Since the trailing matrix update does not have to be applied immediately, but only
to the columns which are factorized, this two-sided update can also be aggregated and used in a left-looking algorithm. 
For instance, to multiply $Q^TXQ$ by a matrix $Y$, we can compute 
{\anlsmall
\begin{align}
Q^TXQY=XY+UV^TY+VU^TY. \label{eq:fctrup}
\end{align}
}

Returning to algorithms that compute a series of $k$ two-sided transformations, we note that when computing $V_2$ from $U_2$ (to apply $Q_2$), we need to multiply $U_2$ by a submatrix of $Q_1^TAQ_1$, which can be done without applying $Q_1$, using the above form.
Left-looking algorithms which generalize this idea and employ a delayed trailing matrix update have been used to reduce directly to tridiagonal form ($b=1$)~\cite{dongarra1989block}.

However, there are disadvantages to reducing the symmetric matrix directly to tridiagonal form, since
it requires that a vector be multiplied by the trailing matrix for each computation of $V_i$ of
which there are $n-2$. These matrix-vector multiplications require $O(n)$ synchronizations 
and $O(n)$ transfers of the trailing matrix between memory and cache (so long as it does not fit into cache).
These disadvantages motivated approaches where the matrix is not reduced
directly to tridiagonal form, but rather to banded form, which allows for $b>1$ Householder vectors
to be computed via QR at each step without needing to touch the trailing matrix from within the QR.
After such a reduction to banded form, it is then necessary to reduce the banded matrix to tridiagonal form.
However, this can be significantly less expensive because the trailing matrix is banded and requires less work and vertical communication to update
than during the full-to-banded reduction step.

Such a multi-stage reduction approach was introduced by~\cite{SBR1,SBR2} with the aim of achieving BLAS 3 reuse. 
These algorithms can reduce the banded matrix to tridiagonal or perform more stages of reduction, employing multiple intermediate band-widths. 
Performing more stages of successive band reduction can improve the synchronization cost of the overall approach, from $O(n)$ as needed if reducing to tridiagonal form directly, to $O(\sqrt{p})$ as shown by~\cite{BDK13-TR}.
ELPA~\cite{Auckenthaler2011272} is a distributed-memory library implementing a two-step reduction approach, motivated by reducing vertical communication cost.
ELPA employs the parallel banded-to-tridiagonal algorithm introduced by~\cite{doi:10.1137/0914078}. 
Performance studies by~\cite{Auckenthaler2011272} have demonstrated that this approach is particularly beneficial for large matrices.

We first introduce an algorithm for reducing a full dense matrix to banded form, with up to $O(p^{1/6})$ less horizontal communication than previously known schemes.
We subsequently introduce an algorithm for reducing a banded matrix to a smaller band-width, again with less communication than known approaches.
Both of these reduction algorithms use a parallel routine `QR', which performs QR factorization and outputs the Householder representation ($U$, $T$) of the $Q$ factor.
We then give a combined, 2.5D symmetric eigensolver algorithm, that uses the first algorithm to reduce the dense symmetric matrix to band-width $\frac{n}{\max(p^{2-3\delta},\log p)}$, then uses $O(\log p)$ calls to our band-to-band reduction, to arrive at a band-width of $n/p$, which is small enough to allow for efficient sequential computation of eigenvalues.
The resulting symmetric eigensolver has the same BSP complexity as QR factorization (Lemma~\ref{lem:parQRsq}), modulo logarithmic factors in the number of processors for the vertical communication and synchronization costs.

  \subsection{Full-to-Band Reduction}
  \label{sec:se:symdirect}

Algorithm~\ref{alg:parDRSE} reduces a symmetric $n$-by-$n$ matrix $A$
to band-width $b$ using replication of data  and aggregation.
It achieves a horizontal communication cost of
$W=O(n^2/p^{\delta}),$
when the amount of available memory on each processor is 
$M = O(n^2/p^{2(1-\delta)}).$
The algorithm is left looking, meaning it updates the next matrix panel (line~\ref{li:parDRSE:upd}) immediately prior to performing the QR of the panel.
Figure~\ref{fig:se_uv} displays the key matrices employed in Algorithm~\ref{alg:parDRSE}, specifically the third and fourth steps of recursion.

The algorithm replicates the matrix $A$ and aggregates as well as replicates the updates $U^{(0)}$ and $V^{(0)}$ (these update matrices should have $m=0$ columns for the initial invocation of Algorithm~\ref{alg:parDRSE}) over $c=p^{2\delta -1}$ layers of $q^2=p^{2(1-\delta)}$ processors. 
In the definition of the algorithm and the analysis we assume that $c$ and $q$ are integers for any given $p$. 
Each of these replicated matrices is stored in a 2D cyclic distribution on each processor grid layer, adhering to the layout assumptions of Algorithm~\ref{alg:StMM}.
A cyclic layout yields local blocks which can be used within sequential routines the same way as done in a blocked layout.
The assumption $b\bmod q=1$ ensures that whenever each new panel of $U$ and $V$ is replicated ($U_1$ and $V_1$ on line \ref{li:UVrep}), they can be concatenated to previously replicated panels while maintaining a perfectly load balanced cyclic distribution.

Algorithm~\ref{alg:parDRSE} performs the update correctly since, first, the computation of $W=\bar{A}U$ where $\bar{A}=Q^TAQ$ (line~\ref{li:parDRSE:formW}) follows the identity Eqn.~\eqref{eq:fctrup}.
Further, as computed on line~\ref{li:parDRSE:formV}, $V$ takes the desired form,
{\anlsmall $$V=\left(\frac{1}{2}UT^TU^T-I\right)WT=\frac{1}{2}UT^TU^T\bar{A}UT-\bar{A}UT,$$}
the same one as the aggregated update matrix derived in Eqn.~\eqref{eq:auvvu}. Consequently, the eigenvalues of the original matrix are preserved in the resulting banded matrix due to the ensured condition on the result of the tail recursion,
which performs the update and factorization of the trailing matrix. In the base case, the matrix
dimension is less than or equal to the desired matrix band-width, which means it suffices to perform
the aggregated update and return the result, which would appear in the lower right block of the
full banded matrix. We now analyze the execution time of Algorithm~\ref{alg:parDRSE} in the BSP cost model. 

\begin{lemma}
\label{lem:parsymdir}
Algorithm~\ref{alg:parDRSE} can reduce any symmetric $n$-by-$n$ matrix (input in any evenly-distributed layout and with $n\geq p$) to a banded matrix with the same eigenvalues and any band-width $n/p^\delta\leq b \leq n/\log p$, using $M=O(n^2/p^{2(1-\delta)})$ memory for any $\delta\in[1/2,2/3]$, when $\cachesize > 3n^2/p^{2(1-\delta)}$, in BSP time,
{\thmsmall
\begin{align*}
O\bigg(&\gamma\cdot \frac{n^3}{p} + \beta\cdot \frac{n^2}{p^\delta} + \alpha\cdot p^\delta\log^2 p \bigg). \end{align*}
}
If $\cachesize \leq 3n^2/p^{2(1-\delta)}$, then there is an additional vertical communication cost of $O(\nu\cdot(n/b)n^2/p^{2(1-\delta)})$.
\end{lemma}

\begin{algorithm}[t]
\caption {$[B]\gets \text{2.5D-Full-to-Band}(A,U^{(0)},V^{(0)},\Pi,b)$}
\label{alg:parDRSE}
{\algsmall
\begin{algorithmic}[1]
\Require {Given nonnegative integers $p,n,m,b$ and $\delta\in[1/2,2/3]$, $z=(bp^\delta/n)^{(1-\delta)/\delta}$: 
$\Pi$ is a grid of $q\times q\times c$ processors where $q=p^{1-\delta}$ and $c=p^{2\delta-1}$ and $b\bmod q=1$,
$A$ is an $n$-by-$n$ symmetric matrix,
$U^{(0)}$ and $V^{(0)}$ are $n$-by-$m$ matrices where $U^{(0)}$ is trapezoidal (zero in top right upper $b$-by-$b$ triangle) and $V^{(0)}$ is dense,
$A$ (stored as a nonsymmetric matrix), $U^{(0)}$, and $V^{(0)}$ are distributed cyclically over $\Pi[\ALL,\ALL,k]$ for each $k\in[1,c]$.}
  \If {$n \leq b$}
  \State Compute $B=A+U^{(0)} {V^{(0)}}^T  +V^{(0)} {U^{(0)}}^T$ and exit.
\vspace{.05cm}
  \EndIf
\State Subdivide $A=\begin{bmatrix} A_{11} & A_{21}^T \\ A_{21} & A_{22}\end{bmatrix}$ where $A_{11}$ is $b$-by-$b$
\State Subdivide $U^{(0)}=\begin{bmatrix} U^{(0)}_1 \\ U^{(0)}_2\end{bmatrix}$ and $V^{(0)}=\begin{bmatrix} V^{(0)}_1 \\ V^{(0)}_2\end{bmatrix}$ where $U^{(0)}_1$ and $V^{(0)}_1$ are $b$-by-$m$

\State Compute $\begin{bmatrix} \bar{A}_{11} \\ \bar{A}_{21} \end{bmatrix} 
    =  \begin{bmatrix} A_{11} \\ A_{21} \end{bmatrix} + U^{(0)}{V^{(0)}_1}^T  +V^{(0)} {U^{(0)}_1}^T$ 
  \label{li:parDRSE:upd}

\vspace{.05cm}
\Comment {Compute QR of matrix panel}
\State $[U_1,T,R] \gets \text{QR}(\bar{A}_{21},\Pi[\ALL,1\TO z,\ALL])$ \label{li:parDRSE:tsqr}
\vspace{.05cm}

\State Compute $W=A_{22}U_1+U_2^{(0)}({V_2^{(0)}}^TU_1)+V_2^{(0)}({U_2^{(0)}}^TU_1)$ \label{li:parDRSE:formW}
\vspace{.05cm}
\State Compute $V_1=\frac{1}{2}U_1(T^T(U^T(WT)))-WT$ \label{li:parDRSE:formV}
\State Replicate $U_1$ and $V_1$ so that they are distributed cyclically over $\Pi[\ALL,\ALL,k]$ for each $k\in[1,c]$ \label{li:UVrep}
  \Comment {Recursively reduce the trailing matrix to banded form}
\State $B_2=\text{2.5D-Full-to-Band}(A_{22},[U_2^{(0)},U_1],[V_2^{(0)},V_1],\Pi,b)$
\State {\small
\(B=\left[\begin{array}{
    >{\centering\arraybackslash$}m{0.3cm}<{$}|
    >{\centering\arraybackslash$}m{0.2cm}<{$}
    >{\centering\arraybackslash$}m{0.2cm}<{$}}
 \hspace{-0.1cm}     \bar{A}_{11} & R^T & 0  \\ \hline
 \hspace{-0.1cm}     R & \multicolumn{2}{>{\centering\arraybackslash$}m{0.6cm}<{$}}{\multirow{2}{*}{$B_2$}} \\ 
 \hspace{-0.1cm}     0 & & \end{array}\right]\)
}
\vspace{.05cm}

\Ensure { $B$ is a symmetric $n$-by-$n$ matrix with band-width $b$ and the same eigenvalues as $A+U^{(0)} {V^{(0)}}^T  +V^{(0)} {U^{(0)}}^T$.
}
\end{algorithmic}
}
\end{algorithm}
\begin{figure}[t] 
\centering
\includegraphics[width=3.in]{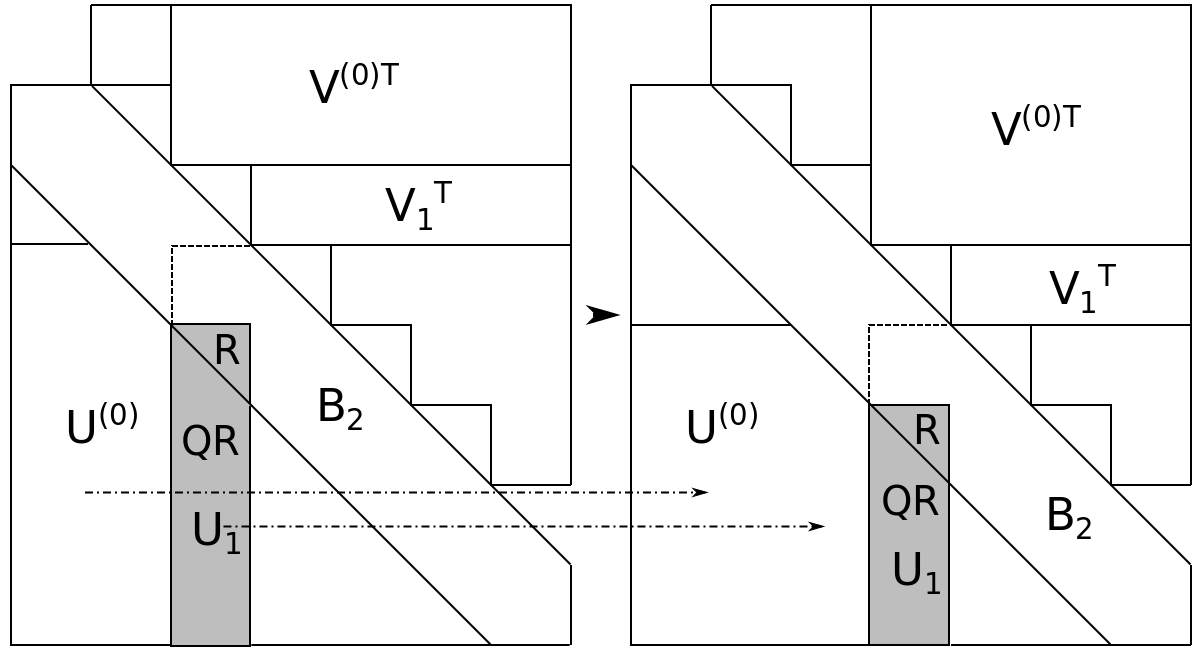}
\caption{A depiction of matrices used in Algorithm~\ref{alg:parDRSE} for two subsequent recursive steps.}
\label{fig:se_uv}
\end{figure}

\begin{proof}
Since $b\geq n/p^\delta$, we assume without loss of generality that $b\mod p^{1-\delta}=0$. 
We also note that since $b\geq n/p^\delta$, $z=(bp^\delta/n)^{(1-\delta)/\delta}\geq 1$.
We note that the dimensions of $A$, $U^{(0)}$, and $V^{(0)}$ at any recursive step will always be less than the dimension of the original matrix, $n$.
Algorithm~\ref{alg:parDRSE} assumes $A$, $U^{(0)}$, and $V^{(0)}$ are initially replicated.
Since each $b\times b$ block of these matrices is distributed cyclically and since $b\bmod q=0$ ($q=p^{1-\delta}$), the submatrix extraction and concatenation done between recursive steps, can preserve perfect load balance without communication.
To satisfy initial assumptions of the first invocation of Algorithm~\ref{alg:parDRSE}, we need to replicate the $A$ matrix.
Since, by assumption, it is distributed over all processors initially, the replication can be done with \(O(n^2/q^2)=O(n^2/p^{2(1-\delta)})\) horizontal communication cost.

At each recursive step, Algorithm~\ref{alg:parDRSE} performs a QR factorization, several matrix multiplications, and replicates $U_1$ and $V_1$.
Each $O(n)\times b$ QR factorization is done using a processor subgrid of dimensions $p^{1-\delta}\times z\times p^{2\delta- 1}$ with a total of $zp^\delta=p(b/n)^{(1-\delta)/\delta}$ processors (picked to minimize both communication and synchronization) using Algorithm~\ref{alg:rectQR}.
By Theorem~\ref{thm:parQRrect} and the fact that $z\geq 1$, it takes BSP time,
{\thmsmall
\begin{align*}
O\bigg(&\gamma\cdot \frac{n^{1/\delta}b^{3-1/\delta}}{p} + \beta\cdot \frac{nb}{p^\delta}
       + \alpha\cdot \frac{b}{n}p^\delta\log^2 p\bigg),
\end{align*}
}
using $M=O\big(\big(\frac{b^\delta n^{1-\delta}}{(zp^\delta)^{1-\delta}}\big)^2\big)=O\big(\big(\frac{n(b/n)^{(2\delta-1)/\delta}}{p^{1-\delta}}\big)^2\big)$ memory.

The two matrix multiplications on line~\ref{li:parDRSE:upd} and the five matrix multiplications on line~\ref{li:parDRSE:formW} (done right to left), all correspond to an $O(n)\times O(n)$ replicated matrix multiplied by an $O(n)\times b$ rectangular matrix.
By Lemma~\ref{lem:repupd}, with $w=\max(1,bp^{2-3\delta}/n)$, using $M=O(n^2/p^{2(1-\delta)}+nb/(wp^{\delta}))=O(n^2/p^{2(1-\delta)})$ memory, the time to compute these matrix multiplications is, if $U^{(0)}$ and $V^{(0)}$ start in cache,
{\mathsmall      $$O\left(\gamma\cdot \frac{n^2b}{p}
          + \beta\cdot \frac{nb}{p^\delta}
          + \alpha\cdot w\right).$$
}
In general (for any cache size), there is an additional cost of 
        $O(\nu\cdot\frac{(n/b)n^2}{p^{2(1-\delta)}})$.
The memory usage needed for these matrix multiplications is greater than that needed for the QR factorizations done by each set of processors.
Since $n^{1/\delta}b^{3-1/\delta}< n^2b$ the computation cost of these matrix multiplications also dominates that of the QR factorizations.

The matrix multiplications needed to compute line~\ref{li:parDRSE:formV} from right to left either operate on an $O(n)\times b$ matrix and a $b\times b$ matrix, like $W\cdot T$, or result in a $b\times b$ matrix, like $U^T\cdot (WT)$.
By Lemma~\ref{lem:parMM} any matrix multiplication where two of the matrix dimensions are $b$ and one is $O(n)$, with $v=p^{2-3\delta}$, takes BSP time,
{\mathsmall \[O\bigg(\gamma\cdot \frac{nb^2}{p} + \beta\cdot \left[\frac{nb}{p} + \frac{n^{2/3}b^{4/3}}{p^\delta}\right]+\alpha\cdot p^{2-3\delta}\log p\bigg).\]}
Since $b\leq n/\log p$, the above communication cost is never greater than that of the larger matrix multiplications, i.e.\ $n^{2/3}b^{4/3}/p^{\delta}\leq nb/p^\delta$.
The synchronization cost of the QR factorizations dominates that of of the matrix multiplications.

Replicating $U_1$ and $V_1$ over $c$ subsets of $q^2$ processors (line~\ref{li:UVrep}) can be done in time,
\(O\lt(\beta \cdot nb/p^{2(1-\delta)}+\alpha\rt)\).

Therefore, the cost over all $n/b-1$ recursive steps when all replicated matrices fit into cache (when $\cachesize>3n^2/p^{2(1-\delta)}$) is the total cost postulated in the theorem.
In the second scenario $\left(\text{when~} \cachesize<3n^2/p^{2(1-\delta)}\right)$, the algorithm incurs an extra additive factor of $O\big((n/b)\frac{wn^2}{p^{2(1-\delta)}}\big)$ in vertical communication cost.
The memory usage is dominated by the replicated matrix multiplication (invocation of Lemma~\ref{lem:repupd} above), which is also as stated in the theorem.
\end{proof}

  \subsection{Band-to-Band Reduction}
  \label{sec:se:bandred}
  
We now consider algorithms for reducing a banded matrix to a smaller band-width, while preserving eigenvalues.
We start by recalling a parallel algorithm designed for small band-widths~\cite{BDK13-TR}, then present
Algorithm~\ref{alg:sebra}, which is designed to exploit additional parallelism given larger starting band-widths.
Algorithm~\ref{alg:sebra} describes the QR factorizations and applications
necessary to reduce a symmetric banded matrix $A$ from band-width $b$ to band-width
$\h=b/k$ via bulge chasing. 
The algorithm eliminates $n/\h$ trapezoidal panels via QR factorization, each of which generate bulges of nonzeros in the trailing matrix.
Each bulge is subsequently chased down the band by $O(n/b)$ eliminations again done by QR factorizations. 
Every new panel elimination is done immediately after the previously generated bulge is chased twice (including its initial panel elimination).
Figure~\ref{fig:bch} depicts the QR factorizations necessary to eliminate a trapezoidal panel and chase two bulges generated from eliminating the first two panels, which are done concurrently in the algorithm.
This type of pipelined successive band reduction approach was first considered by~\cite{SBR1,SBR2}.
The CA-SBR algorithm in~\cite{BDK13-TR} is similar, but assigns each processor a set of bulge chases at each pipeline step, rather than performing each bulge chase with a set of processors as done in Algorithm~\ref{alg:sebra}.

\begin{lemma}
\label{lem:sbr_topc}
An $n\times n$ symmetric matrix (input in any load-balanced layout) of band-width $b\leq n/p$ can be reduced to one with the same eigenvalues and band-width $b/2$, using $M=O(nb/p)$ memory, in BSP time,
{\thmsmall 
\begin{align*}
O\bigg(&\gamma\cdot \frac{n^2b}{p} + \beta\cdot nb +\nu\cdot \frac{n^2}{p}+ \alpha\cdot p \bigg).
\end{align*}
}
\end{lemma}
\begin{proof}
We consider the cost of one step of the CA-SBR algorithm~\cite{BDK13-TR}.
A redistribution from any initial layout costs $O(\beta\cdot nb +\alpha)$.
The analysis in~\cite{BDK13-TR} shows that the cost of reducing from bandwidth $b$ to $b/2$ has the computation, horizontal communication, and synchronization costs, as well as the memory usage postulated in the lemma.
The algorithm consits of a bulge chase pipeline, executed in $O(p)$ parallel steps, in which each processor works on $O(n/p)$ columns, chasing $O(n/(pb))$ bulges $O(n/(pb))$ times, for a total of $O(n^2/(p^2b^2))$ bulge chases.
Since each bulge chase consists of a QR factorization and a matrix multiplication, with matrices of size $O(b)\times O(b)$, by Lemma~\ref{lem:seqMM} and Lemma~\ref{lem:seqqr}, the vertical communication cost is $O(\nu\cdot b^2)$ for each bulge chase.
Summing the costs of the bulge chases over all parallel steps yields the postulated total cost.
\end{proof}

We now consider the cost of Algorithm~\ref{alg:sebra}. 
Its primary innovation is to perform each QR factorization and update in parallel using a subset of processors, leveraging both pipelined parallelism across different bulge chases as well as parallelism within a bulge chase.

\begin{lemma}
\label{lem:sbr}
Algorithm~\ref{alg:sebra} can reduce an $n\times n$ symmetric matrix (input in any evenly-distributed layout) of band-width $b\geq n/p$ to one with the same eigenvalues and band-width $b/k$, using $M=O((n^{1-\delta}b^{\delta}/p^{1-\delta})^2)$ memory for any $\delta\in[1/2,2/3]$ and any $k\leq 1+p^{2-3\delta}$, in BSP time,
{\thmsmall
\begin{align*}
O\bigg(&\gamma\cdot \frac{n^2b}{p} + \beta\cdot \frac{n^{1+\delta}b^{1-\delta}}{p^\delta} + \alpha\cdot \frac{k^\delta n^{1-\delta}p^\delta}{b^{1-\delta}}\log{p} \bigg).
\end{align*}
}
\end{lemma}

\begin{proof}
The cost of each inner loop iteration (loop on line~\ref{li:2.5D-SBR:forj}) can be derived from the costs of the matrix multiplications and QR done inside it.
Let the pair $(i,j)$ correspond to the the $i$th iteration of the outer loop and $j$th iteration of the inner loop.
Figure~\ref{fig:bch} displays the QR factorizations and updates computed during a few such iterations.
Each iteration computes a QR factorization of a matrix with dimensions at most $(b-\h) \times \h$, $B[I_\mathrm{qr.rs}, I_\mathrm{qr.cs}]$ on 
line~\ref{li:2.5D-SBR:QR} with $\bar{p}=pb/(nk^{(1-\delta)/\delta})$ processors.
The BSP time to compute such a QR factorization is by Theorem~\ref{thm:parQRrect} for $\delta\in [1/2,2/3]$,
{\mathsmall \begin{align*}
O\bigg(&\gamma\cdot \frac{b\h^2}{\bar{p}}
        + \beta\cdot \frac{b^\delta \h^{2-\delta}}{\bar{p}^\delta} 
    + \alpha\cdot \bar{p}^\delta\log(\bar{p}) \bigg) \\
=
O\bigg(&\gamma\cdot \frac{nb^2}{k^{3-1/\delta}p}
        + \beta\cdot \frac{n^\delta b^{2-\delta}}{kp^\delta} 
    + \alpha\cdot k^{\delta-1}(pb/n)^\delta\log{p} \bigg)
\end{align*}}
The amount of memory needed for this QR factorization is given in Lemma~\ref{thm:parQRrect} as $M=O\big((\h^\delta b^{1-\delta}/\bar{p}^{1-\delta})^2\big)=O((n^{1-\delta}b^{\delta}/(p^{1-\delta}k^{(2\delta -1)/\delta}))^2)$.

The matrix multiplications to form the $V$ matrix are on lines \ref{li:2.5D-SBR:formW} and \ref{li:2.5D-SBR:formV}, while those to perform the updates are on lines \ref{li:2.5D-SBR:updA1} and \ref{li:2.5D-SBR:updA2}.
The matrix multiplications on line \ref{li:2.5D-SBR:formV} should be done from right to left.
We can then observe that the most costly matrix multiplications in Algorithm~\ref{alg:sebra} are $B[I_\mathrm{up.cs}, I_\mathrm{qr.rs}]U$ on line~\ref{li:2.5D-SBR:formW} and the updates $UV^T$ and $VU^T$ on lines \ref{li:2.5D-SBR:updA1} and \ref{li:2.5D-SBR:updA2}.
In the first case, a $(3b-\h)\times(b-\h)$ is multiplied by a $(b-\h)\times \h$ matrix, while the update $UV^T$ involve $(3b-\h)\times \h$ matrix multiplied by an $\h \times (b-\h)$ matrix ($VU^T$ is just the transpose of the former).
In both cases, by Lemma~\ref{lem:parMM} with $v=\hat{p}^{2-3\delta}/(k-1)$ (we subtract one from $k$ to make sure $v\geq 1$), the BSP time to compute the matrix multiplications using $\hat{p}$ processors is
{\mathsmall 
\begin{align*}
    O\bigg(&\gamma\cdot \frac{b^2h}{\hat{p}} 
    + \beta\cdot\frac{b^2}{k\hat{p}^\delta} 
    + \alpha\cdot \frac{\hat{p}^{2-3\delta}}{k}\log p  \bigg) \\
=    O\bigg(&\gamma\cdot \frac{nb^2}{kp} 
    + \beta\cdot\frac{n^\delta b^{2-\delta}}{kp^\delta} 
    + \alpha\cdot \frac{(pb/n)^{2-3\delta}}{k}\log p  \bigg),
    \end{align*}
}
with a memory footprint of $M=O(b^2/\hat{p}+(b^2h/(v\hat{p}))^{2/3})=O((b/\hat{p}^{1-\delta})^2)=O((n^{1-\delta}b^{\delta}/p^{1-\delta})^2)$, which is greater than the memory needed to perform the QR factorizations.
The other matrix multiplications have strictly lower cost and the cost of redistributions necessary for all of these matrix multiplications is included in the horizontal communication cost of Lemma~\ref{lem:parMM}.
As $A$ and $B$ are stored in load balanced layouts, each processor subset can obtain the submatrix which it factorizes and the submatrix which it updates at every iteration with $O(b^2/\hat{p})$ horizontal communication.

\begin{algorithm}[t]
{\algsmall
\caption {$[B]\gets 
            \text{2.5D-Band-to-Band}(A,\Pi,b,k)$}
\label{alg:sebra}
\begin{algorithmic}[1]
\Require {Given positive integers $b,p,n,k$ and $\h=b/k$ 
           with $n\bmod b \equiv 0$ and $b\bmod k \equiv 0$:
           $A$ is a banded symmetric matrix of dimension $n$ with band-width $b\leq n$,
                       $\hat{\Pi}_j\subset \Pi$ is
           the $j$th group of $\hat{p}\equiv pb/n$ processors for $j\in\inti{1}{n/b}$.}
\State Set $B=A$
\State Let $B[(j-1)b+1 \TO jb\CMA (j-1)b+1\TO jb]$ be replicated in $\hat{\Pi}_j$
  over $(bp/n)^{2\delta -1}$ subsets of  $(bp/n)^{2(1-\delta)}$ processors.
\Comment{Iterate over panels of $B$}
\For{ $i \in \inti{1}{n/h-1}$ } \label{li:2.5D-SBR:fori}
          \Comment {$\hat{\Pi}_j$ applies chase $j$ of bulge $i$ as soon as $\hat{\Pi}_{j-1}$ executes chase $(j-1)$}
  \For {$j=1:\lfloor(n-i\h-1)/b\rfloor$} \label{li:2.5D-SBR:forj}
    \Comment{Define row and column offsets}
    \State Let $o_\mathrm{blg}=(i-1)\h+(j-1)b, \ \  o_\mathrm{qr.r}=o_\mathrm{blg}+\h$
    \If {$j=1$} $o_\mathrm{qr.c}=o_\mathrm{qr.r}-\h, \ \ o_\mathrm{v} = 0$
    \Else \ $o_\mathrm{qr.c}=o_\mathrm{qr.r}-b, \ \ o_\mathrm{v} = b-\h, \ \ o_\mathrm{up.c} = o_\mathrm{qr.c}+\h$ \EndIf
    \Comment{Define index ranges needed for bulge chase}
    \State $n_\mathrm{r}= \min(n-o_\mathrm{qr.r}, b),  n_\mathrm{c}= \min(n-o_\mathrm{up.c}, \h+3b)$
    \State $I_\mathrm{qr.rs}=o_\mathrm{qr.r}+(1\TO n_\mathrm{r}), \ \ I_\mathrm{qr.cs}=o_\mathrm{qr.c}+(1\TO \h)$
    \State $I_\mathrm{v.rs}=o_\mathrm{v}+(1:n_\mathrm{r}), \ \ I_\mathrm{up.cs}=o_\mathrm{up.c}+(1\TO n_\mathrm{c})$
    \Comment { Perform a rectangular parallel QR factorization }
    \State $[U,T,R] \gets \text{QR}(B[I_\mathrm{qr.rs}, I_\mathrm{qr.cs}],\hat{\Pi}_j[1\TO p\h/n])$ \label{li:2.5D-SBR:QR}
    \State  $B[I_\mathrm{qr.rs}, I_\mathrm{qr.cs}]=\begin{bmatrix} R \\ 0 \end{bmatrix}$, 
               $B[I_\mathrm{qr.cs}, I_\mathrm{qr.rs}]=\begin{bmatrix} R \\ 0 \end{bmatrix}^T$
    \Comment { Perform trailing matrix updates }
    \State  $W=B[I_\mathrm{up.cs}, I_\mathrm{qr.rs}]UT, \quad V=-W$ \label{li:2.5D-SBR:formW}
    \State  $V[I_\mathrm{v.rs}, \ALL]=V[I_\mathrm{v.rs}, \ALL]+ \frac{1}{2}U(T^T(U^TW[I_\mathrm{v.rs}, \ALL]))$ \label{li:2.5D-SBR:formV}
    \State  $B[I_\mathrm{qr.rs}, I_\mathrm{up.cs}]=B[I_\mathrm{qr.rs}, I_\mathrm{up.cs}]+UV^T$ \label{li:2.5D-SBR:updA1}
    \State  $B[I_\mathrm{up.cs} , I_\mathrm{qr.rs}]=B[I_\mathrm{up.cs} , I_\mathrm{qr.rs}]+VU^T$ \label{li:2.5D-SBR:updA2}
  \EndFor
\EndFor

\Ensure { $B$ is a banded matrix with band-width $h$ and the same eigenvalues as $A$ } \end{algorithmic}
}
\end{algorithm}
\begin{figure}[t] 
\centering
\includegraphics[width=3.3in]{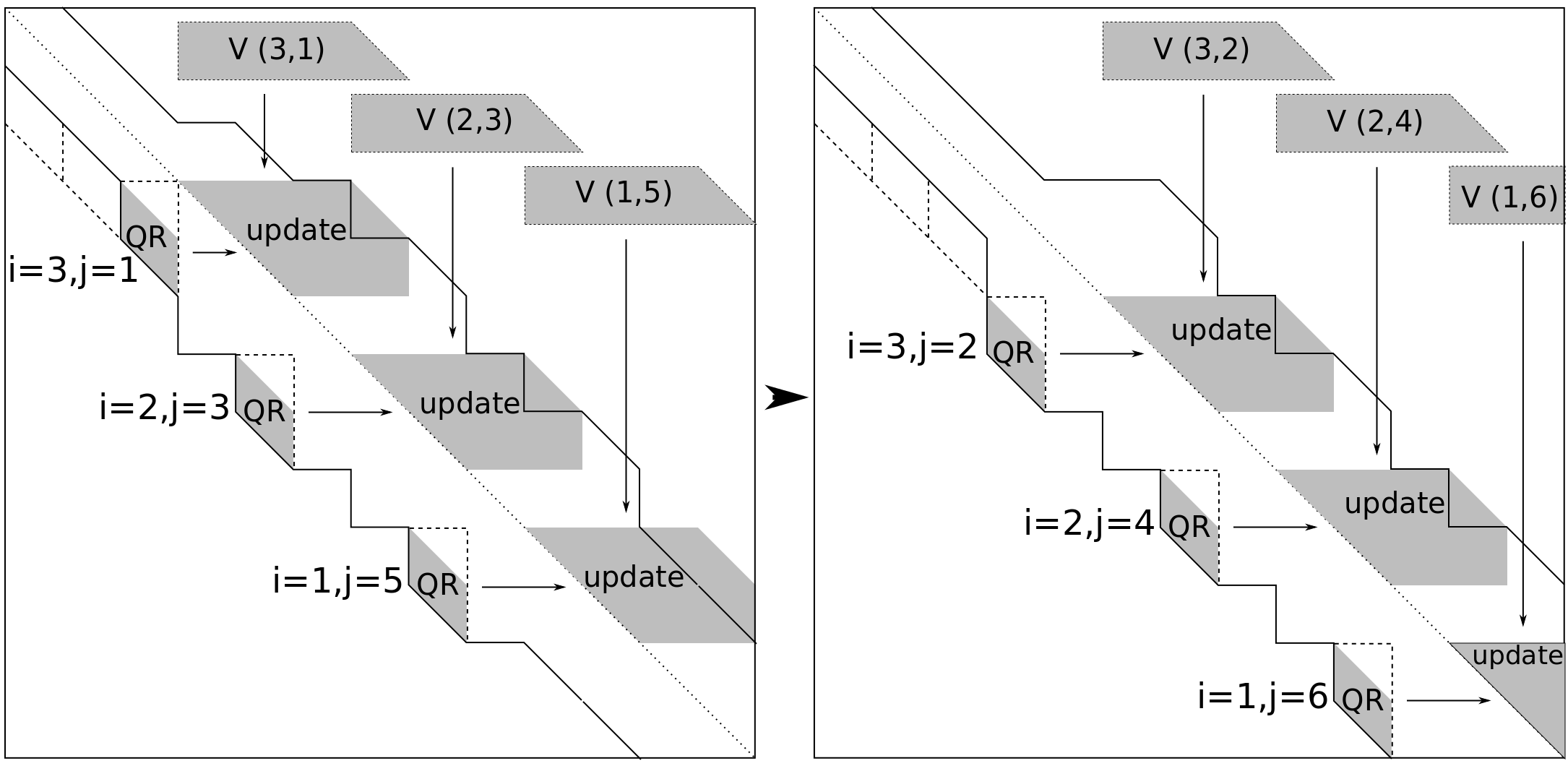}
\caption{QR factorizations and updates in iterations $(i,j)\in\{(3,1),(2,3),(1,5)\}$ (left) and $(i,j)\in\{(3,2),(2,4),(1,6)\}$ (right) of Algorithm~\ref{alg:parDRSE} with $k=2$. These two sets of iterations are executed concurrently by processor groups $\hat{\Pi}_1$, $\hat{\Pi}_3$, and $\hat{\Pi}_5$ (left) and  $\hat{\Pi}_2$, $\hat{\Pi}_4$, and $\hat{\Pi}_6$ (right), respectively. Only the unique part of the trailing matrix update is shown, while the pseudocode performs both symmetric reflections of it. Each matrix $V$ is labeled with the iteration in which it is computed.}
\label{fig:bch}
\end{figure}

Thus, the overall cost for each iteration of Algorithm~\ref{alg:sebra} is the sum of the two different costs above,
    {\mathsmall\begin{align*}
    O\bigg(&\gamma\cdot \frac{nb^2}{kp} 
    + \beta\cdot\frac{n^\delta b^{2-\delta}}{kp}  
    + \alpha\cdot k^{\delta-1}(pb/n)^\delta\log{p} \bigg).
    \end{align*}}
For a given outer loop (line~\ref{li:2.5D-SBR:fori}) iteration $i$, each $j$ loop iteration (line~\ref{li:2.5D-SBR:forj}) is done by a different processor group.
The total number of inner loop iterations is roughly $(n/h)(n/b)/2$ and they are pipelined among $n/b$ groups of processors, up to $n/(2b)$ of them working concurrently on different bulge chases at any given time.
Consequently, the algorithm can be executed in $O(n/h)$ phases, where at the $i$th phase, $\min(i-1,(n-i\h)/(2b))$ processor groups chase bulges concurrently and the $i$th panel is eliminated.
At each phase, a synchronization and data exchange is required between the QR factorization and trailing matrix updates computed by adjacent active processor groups.
Therefore, the BSP cost of each recursive step of the algorithm corresponds to the cost of computing $O(n/\h)=O(kn/b)$ inner loop iterations using one processor group, which corresponds to the cost postulated in the lemma.
\end{proof}

  \subsection{Complete Symmetric Eigensolver}
  \label{sec:se:sucred}

Algorithm~\ref{alg:25dse} combines our algorithms for full-to-band reduction (Algorithm~\ref{alg:parDRSE}) with multiple subsequent stages of band-to-band reduction (Algorithm~\ref{alg:sebra}) and band-halving steps of the CA-SBR algorithm from~\cite{BDK13-TR}, which we refer to as CA-BR.
Algorithm~\ref{alg:parDRSE} reduces the symmetric matrix to one with band-width at most $n/\log p$.
Algorithm~\ref{alg:sebra} is then used to successively half the band-width to $n/p^\delta$.
Subsequently, the CA-BR algorithm (same function signature as 2.5D-Band-to-Band) is used to reduce the band-width to $n/p$.
At that point, the matrix is small enough for one processor to compute the eigenvalues efficiently.

For every 2.5D-Band-to-Band step that reduces the band-width by a factor of $k$, Algorithm~\ref{alg:sebra} reduces the number of processors used by $k^\zeta$ where $\zeta= (1-\delta)/\delta$. The parameter $\zeta$ is chosen to be $(1-\delta)/\delta$ in order to keep the per-stage horizontal cost term $O(nb/p^\delta)$ from increasing at each recursive step, since $n(b/k)/(p/k^\zeta)^\delta=nb/p^\delta$.
Decreasing the number of active processors in this way also keeps the synchronization cost equal at every stage. 
Overall, we now obtain a parallel algorithm that has horizontal communication of $O(n^2/p^\delta)$, vertical communication of $O(n^2\log p/p^\delta)$, and $O(p^\delta \log^2 p)$ synchronizations.
Modulo logarithmic cost factors in vertical communication and synchronization, this amounts to the same communication cost as the best known algorithms for LU and QR factorization~\cite{snirmatmul,Tiskin2007179,SD_EUROPAR_2011}.

\begin{algorithm}[t]
{\algsmall
\caption {$[D]\gets 
            \text{2.5D-Symmetric-Eigensolver}(A,\Pi)$}
\label{alg:25dse}
\begin{algorithmic}[1]
\Require {Given positive integers $p$, $n$, and $\delta \in [1/2,2/3]$
           with $n\bmod b \equiv 0$,
           $A$ is a symmetric matrix of dimension $n$.}
\State Let $b=\frac{n}{\max(p^{2-3\delta},\log p)}$, $k=2$, and $\zeta = (1-\delta)/\delta$
\State Set $B=A$
\State Execute $B=\text{2.5D-Full-to-Band}(A,\{\},\{\},\Pi,b)$
\For{ $i =0\TO \log_2(bp^\delta/n)-1$ } 
  \State Let $\bar{\Pi}= \Pi[1\TO p/k^{i\zeta}]$
  \State Gather $B$ onto $\bar{\Pi}$
  \State Execute $B=\text{2.5D-Band-to-Band}(B,b/k^{i},\bar{\Pi},k)$
\EndFor
\State Let $\bar{\Pi}= \Pi[1\TO p^\delta]$
\For{ $i =0\TO \log_2(p^{1-\delta})-1$ } 
  \State Execute $B=\text{CA-BR}(B,n/(p^\delta k^{i}),\bar{\Pi},k)$
\EndFor
\State Gather $B$ onto a processor and compute its eigenvalues $D$
\Ensure { $D$ is a vector containing the eigenvalues of $A$}
\end{algorithmic}
}
\end{algorithm}

\begin{theorem}
\label{thm:parse}
Algorithm~\ref{alg:25dse} computes the eigenvalues of a symmetric $n$-by-$n$ matrix (input in any evenly-distributed layout), using $M=O(n^2/p^{2(1-\delta)})$ memory for any $\delta\in[1/2,2/3]$, in BSP time,
{\thmsmall
\begin{align*}
O\bigg(&\gamma\cdot \frac{n^3}{p} + \beta\cdot \frac{n^2}{p^\delta} + \nu\cdot\frac{n^2\log p}{p^\delta}+ \alpha\cdot p^\delta\log^2{p} \bigg).
\end{align*}
}
\end{theorem}
\begin{proof}
The cost of the gather/redistribution of $B$ onto $\bar{\Pi}$ is dominated by the subsequent 2.5D-Band-to-Band invocation.
The cost of computing the eigenvalues of $B$ sequentially at the end is $O(\gamma\cdot n^3/p+\beta\cdot n^2/p+\alpha)$, since the band-width is $n/p$~\cite{BDK13-TR}.
We employ Lemma~\ref{lem:parsymdir} with $b=\frac{n}{\max(p^{2-3\delta},\log p)}$ to obtain the cost of 2.5D-Full-to-Band.
The computation, horizontal communication, and synchronization costs are the same for the call to 2.5D-Full-to-Band as the overall costs postulated in Theorem~\ref{thm:parse}.
The vertical communication cost term incurred for small cache sizes, $O(\nu\cdot(n/b)n^2/p^{2(1-\delta)})$ is bounded by $O(\nu\cdot[ n^2/p^\delta+n^2\log p/p^{2/3}])=O(\nu\cdot n^2\log p/p^\delta)$. 
We now consider the memory footprint and cost of the invocations of 2.5D-Band-to-Band.
By Lemma~\ref{lem:sbr} with $k=2$, the memory usage is $M=O((n^{1-\delta}\bar{b}^{\delta}/\bar{p}^{1-\delta})^2)$,
where $\bar{b} = b/k^i$ where $\bar{p}=p/k^{i\zeta}$ at iteration $i$.
We observe that $(n^{1-\delta}\bar{b}^{\delta}/\bar{p}^{1-\delta})^2=O(n^2/p^{2(1-\delta)})$ for all iterations $i$, because at each subsequent iteration $\bar{b}$ decreases by $k$ while $\bar{p}$ decreases by $k^\zeta$, and so $\bar{b}^{\delta}/\bar{p}^{1-\delta}\leq b^\delta/p^{1-\delta}\leq n^{\delta}/p^{1-\delta}$ for all $i$, since
\(k^{(1-\delta)\zeta}/k^{\delta }=k^{(1-\delta)^2/\delta}/k^{\delta}\leq k^{1-\delta}/k^{\delta}\leq 1.\)
The cost of each band reduction with starting band-width $\bar{b}$ and $\bar{p}$ processors is by Lemma~\ref{lem:sbr} with $k=2$,
{\mathsmall \[O\bigg(\gamma\cdot \frac{n^2\bar{b}}{\bar{p}} + \beta\cdot \frac{n^{1+\delta}\bar{b}^{1-\delta}}{\bar{p}^\delta} + \alpha\cdot \frac{n^{1-\delta}\bar{p}^\delta}{\bar{b}^{1-\delta}}\log{p} \bigg).\]}
The computation cost clearly decreases with each iteration $i$.
The horizontal communication cost is $O(nb/p^\delta)=O(n^2/(p^\delta \log p))$ (since $b\leq n/\log p$) at each iteration, since
{\mathsmall \[\frac{\bar{b}^{1-\delta}}{\bar{p}^{\delta}} = \frac{(b/k^i)^{1-\delta}}{(p/k^{i\zeta})^\delta}=\frac{b^{1-\delta}}{p^\delta}.\]}
Therefore, over all $O(\log p)$ iterations, the bandwidth cost of the SBR invocations is $O(n^2/p^\delta)$.
Finally, the synchronization cost is $\frac{n^{1-\delta}\bar{p}^\delta}{\bar{b}^{1-\delta}}\log{p}=O(p^\delta\log{p})$ at each iteration, since $\bar{p}^\delta/\bar{b}^{1-\delta}=p^\delta/b^{1-\delta}$.
Thus, the overall synchronization cost is bounded by the cost postulated in the theorem.

The time to execute CA-BR using $p^\delta$ processors starting from band-width $n/p^\delta$ and reducing it to $n/p$ is via Lemma~\ref{lem:sbr_topc}, 
$O(\gamma\cdot \frac{n^3}{p^{2\delta}} + \beta\cdot \frac{n^2}{p^\delta} +\nu\cdot \frac{n^2\log p}{p^\delta}+ \alpha\cdot p^\delta \log p)$.
\end{proof}
A disadvantage of this multi-stage approach arises when eigenvectors are required in addition to eigenvalues.
The cost of the back-transformations scales linearly with the number of band-reduction stages (each stage requires $O(n^2)$ memory and $O(n^3)$ computation).
We leave the consideration of eigenvector construction for future work.
To reduce the number of band-reduction stages when $\delta <2/3$, one can use $k=p^{2-3\delta}$ with each invocation of 2.5D-Band-to-Band, but this results in a greater synchronization cost.
It may also be possible to improve the costs of the 2.5D-Band-to-Band algorithm, by using aggregation as in the 2.5D-Full-to-Band algorithm.

\section{Conclusion}
\label{sec:conc}

\begin{table}
{\footnotesize

\centering

\begin{tabular}{|l|c|c|c|}
\hline
Algorithm & $W$ ($\beta$) & $Q$ ($\nu$) & $S$ ($\alpha$)  \\
\hline \\ [-2.4ex] 
ScaLAPACK~\cite{SCALAPACK} & $n^2/\sqrt{p}$ & $n^3/p$ & $n\log p$ \\
ELPA~\cite{auckenthaler2012highly} & $n^2/\sqrt{p}$ & - & $n\log p$ \\
CA-SBR~\cite{BDK13-TR} & $n^2/\sqrt{p}$ & $n^2\log n/\sqrt p$ & $\sqrt{p}(\log^2p+\log n)$ \\
Theorem~\ref{thm:parse}& $n^2/p^{\delta}$ & $n^2\log p/p^{\delta}$ & $p^{\delta}\log^2 p$ 
\\ 
\hline
\end{tabular}
\vspace{.05cm}
}
\caption{Asymptotic communication costs for computing eigenvalues (with $\delta \in[1/2,2/3]$). All variants require $O(n^3/p)$ computation.}
\label{tab:compse}
\end{table}

Table~\ref{tab:compse} provides a comparison of communication and synchronization costs to previous work.
Our new direct method for computing the eigenvalues of a symmetric matrix, performs up to $p^{1/6}$ less horizontal communication than alternatives.
The vertical communication cost ($Q$) for ScaLAPACK assumes $H<n^2/p$ and arises from the matrix-vector multiplications computing $V$ for each column.
For CA-SBR, $Q$ is inferred from Lemma~\ref{lem:sbr_topc}.
For ELPA, we assume the full-to-band step reduces to band-width $b=\sqrt{H}$, in which case either (when $\sqrt{H}>n/p$) the banded matrix fits in cache, or
$\nu\cdot Q=O(\nu\cdot [n^3/(pb)+nb^2])= O(\gamma\cdot F/\sqrt{H})$~\cite{auckenthaler2012highly}.

The new 2.5D-Symmetric-Eigensolver algorithm trades off a variable amount of extra work, synchronization, and memory usage for a lower communication cost.
Implementations of the algorithms in this paper permit optimizations such as
\begin{itemize}
\item alternating between left-looking partial updates and complete trailing matrix updates in Algorithm~\ref{alg:parDRSE},
\item smaller bulge width in Algorithm~\ref{alg:sebra} to increase parallelism in the bulge chase pipeline,
\item lookahead~\cite{agarwal1988parallel,strazdins1998comparison} (overlapping QR with updates).
\end{itemize}

Our analysis shows that a carefully parameterized collage of parallel algorithms and optimizations yields asymptotic cost improvements with minimal overhead.
We combine approaches (2.5D algorithms, aggregation, successive band reduction) that have been successful on modern architectures~\cite{SBD_SC_2011,BDGJNS_IPDPS_2014,Auckenthaler2011272}, so our innovations should pave the path for practical improvements in scalability of applications computing singular values or eigenvalues of matrices.

\bibliographystyle{IEEEtran}
\bibliography{paper}

\end{document}